\documentclass{svproc}
\usepackage{url}
\usepackage{amssymb}
\usepackage{amsmath}

\usepackage{graphicx}
\usepackage{xcolor}
\usepackage{algorithm}
\usepackage{hyperref}
\newtheorem{assumption}{Assumption}

\begin{document}
\mainmatter              
\title{Multi-Objective Learning Model Predictive Control}
\titlerunning{Multi-Objective Learning Model Predictive Control}  
%
\author{Siddharth H. Nair\thanks{Contributed equally.}\inst{1}, Charlott Vallon$^\star$\inst{1}, and Francesco Borrelli\inst{1}}
\authorrunning{Nair et al.} 
%
%
\institute{University of California, Berkeley, Berkeley CA 94720, USA,\\
\email{charlottvallon@berkeley.edu}}
\maketitle              

\begin{abstract}
Multi-Objective Learning Model Predictive Control is  a novel data-driven control scheme which improves a linear system's closed-loop performance with respect to several convex control objectives over iterations of a repeated task. 
At each task iteration, collected system data is used to construct terminal  components of a Model Predictive Controller. 
 The formulation presented in this paper ensures that closed-loop control performance improves between successive iterations with respect to each objective.
We provide proofs of recursive feasibility and performance improvement, and show that the converged policy is Pareto optimal.
 Simulation results demonstrate the applicability of the proposed  approach. 
\end{abstract}

%
\section{Introduction}
Iterative learning control (ILC) considers the problem of solving the same task multiple times (consider robotic manipulators performing repetitive motions \cite{meng2017iterative}, drone racing \cite{lv2023autonomous}, or batch processes \cite{geng2023data}). Each task execution is referred to as an iteration.
In \textit{data-driven} ILC, system data collected during an iteration is used to improve the control performance at subsequent iterations. 
This approach can be particularly helpful if designing an optimal controller a priori is difficult, for example due to modeling challenges or task complexity. 

The Learning Model Predictive Control (LMPC) method in \cite{rosolia2019learning} proposes an MPC controller formulation that iteratively improves the system's closed-loop performance on the task with respect to a chosen control objective while guaranteeing constraint satisfaction during the learning process. 
Here we extend the LMPC framework to the case where performance is measured with respect to multiple distinct control objectives. 
Most real-world control applications require balancing multiple objectives, such as reducing actuation costs and completing a necessary task, or considering both individual and shared objectives in multi-agent systems \cite{CAMISA20202684}. Designing control formulations that are both implementable and result in the desirable closed-loop behavior can be difficult, especially if there are competing objectives, e.g. minimize both time to complete task and energy required to complete task \cite{taylor2010},\cite{9654900}, \cite{stieler2022}.
In a number of examples, carefully tuned MPC cost function formulations have been proposed that balance multiple application-specific objectives, including for power converters \cite{6565396}, a distillation column \cite{WOJSZNIS2007351}, vehicle adaptive cruise control \cite{zhao2017}, and efficient microgrid management \cite{VASQUEZ2023120998}. 
These works do not include proofs of stability or feasibility; rather, the control parameters are formulated as a function of application-specific values and are validated in simulations or experiments. 

In multi-objective MPC, input trajectories are evaluated with respect to multiple objective functions. In general, a single trajectory which minimizes all costs at the same time does not exist.  
Instead, there will be a Pareto frontier of solutions which do not dominate each other \cite{boyd2004convex}. 
A common approach for application-agnostic multi-objective MPC frameworks is to execute the following steps at each time step: \textit{i)} calculate the Pareto frontier based on the current system and environment state, \textit{ii)} select a specific point on the Pareto frontier, and \textit{iii)} apply the corresponding input. 
Because calculating the Pareto frontier of a multi-objective optimization problem can be difficult and computationally intensive, a number of methods have been proposed to increase efficiency \cite{7074765}, \cite{9655125}, \cite{PEITZ20178674}, including only calculating the relevant subset of the Pareto frontier at each time step \cite{stieler2019}.

In contrast with the reviewed works, we propose using an iterative, data-driven approach for reducing multiple objectives over each iteration of a repeated task. 
In our framework, the MPC controller does not have full knowledge of control invariant sets typically used for guaranteeing stability; instead, this set is constructed iteratively beginning from a single known feasible trajectory.
Critically, no individual closed-loop control cost may be allowed to increase between subsequent iterations. 
Such behavior is desirable in a number of situations. 
Consider, for example, the cases presented in \cite{vallon2024learning}, \cite{vallon2024learning2}, where a high-level planner assigns iterative low-level tasks to an agent.
High-level task assignment occurs based on estimates of the agent's capabilities; for safe planning, the low-level controller must guarantee that as the agent performs a specific task multiple times, incurred control costs do not increase. 
Compared to the capacity-constrained LMPC implementation in \cite{vallon2024learning2}, which minimizes consumption of a single capacity without resulting in increased consumption of other capacities, here we actively minimize multiple objectives. 
In this paper, we:
\begin{enumerate}
    \item propose an application-agnostic Multi-Objective LMPC (MO-LMPC) formulation for iterative closed-loop performance improvement with respect to multiple objectives,
    \item provide proofs of stability and iterative performance improvement with respect to each individual objective, 
    \item prove that if the MO-LMPC converges to a steady-state trajectory, the converged policy is Pareto optimal, and
    \item demonstrate the effectiveness of the formulation in a simulation example.
\end{enumerate}

\textit{ \textbf{Running Example}: Consider a fleet of autonomous electric delivery vehicles. Each day, a high-level scheduler assigns a number of deliveries to each vehicle, which must occur within specific time windows as expected by clients. Task assignment is done on the basis of time and energy cost estimates for routes between delivery locations. The vehicles begin the day fully charged at a centralized depot, and must return there after completing their deliveries without running out of charge.}

\textit{We will refer to this running example throughout the remainder of this paper, to help explain concepts and motivate specific formulations. In this paper we focus on the design of the low-level autonomous vehicle controllers. We refer to \cite{vallon2024learning2} for an example formulation of the high-level scheduler.}

\textit{Notation:}
Unless otherwise indicated, vectors are column vectors constructed using parentheses $(~)$. 
$[A, B]$ denotes the set of integers $\{A, A+1, \dots, B-1, B \}$.



\section{Problem Formulation}\label{sec:PD}

We consider the linear discrete-time system
\begin{align}\label{eq:sysdyn_student}
    x_{t+1}=Ax_t+Bu_t,
\end{align}
 where $x_t\in \mathbb{R}^{n_x}$ and $u_t\in\mathbb{R}^{n_u}$  are the system state and input respectively  at time step $t$. 
 The states and inputs are subject to polytopic constraints given by 
 \begin{subequations}
 \begin{align}
     &\mathcal{X}=\{x\in\mathbb{R}^{n_x}~|~ G_xx\leq g_x\}, ~~\mathcal{U}=\{u\in\mathbb{R}^{n_u}~|~ G_uu\leq g_u\}.
 \end{align}      
 \end{subequations}

We seek feedback policies $\boldsymbol{\pi}=\{\pi_t(\cdot)\}_{t\geq0}$ to be applied to \eqref{eq:sysdyn_student} that optimize the performance of the closed-loop system with respect to $M$ different control objectives $V_i$, where:
\begin{align}\label{vdef}
    V_i(x_0, \boldsymbol{\pi})&=\sum_{t=0}^\infty h_i(x_t,\pi_t(x_t)),~~~ i\in[1,M].
\end{align}    
The $i$th control objective $V_i(x_0, \boldsymbol{\pi})$ is characterized by the stage cost $h_i(\cdot, \cdot)$ incurred over the state and input trajectory that results from applying policy $\boldsymbol{\pi}$ to the system \eqref{eq:sysdyn_student} beginning at state $x_0$ at time $t=0$. We consider stage costs $h_i(\cdot, \cdot)$ which are continuous, jointly convex, and satisfy
\begin{align}\label{eq:stagecost}
    &h_i(x_t, u_t) \succeq 0 ~~~\forall x_t \in \mathbb{R}^{n_x} ~\backslash~ x_F, ~u_t \in \mathbb{R}^{n_u} ~\backslash ~u_F,  ~~  h_i(x_F, u_F) = 0~~ \forall i \in [1,  M],\nonumber\\
    & \exists ~i \in [1,M]~ : ~ h_i(x_t, u_t) \succ 0 ~~~\forall x_t \in \mathbb{R}^{n_x} \backslash x_F, ~u_t \in \mathbb{R}^{n_u} \backslash u_F
\end{align}
for some goal state $x_F \in \mathcal{X}$ which is an equilibrium state under input $u_F \in \mathcal{U}$ so that \eqref{eq:sysdyn_student}:
\begin{align}\label{eq:xf}
    x_F = Ax_F + Bu_F.
\end{align}

\textit{\textbf{Running Example:} 
Each route between two specific delivery locations will be associated with a particular controller, to be applied to a vehicle with state $x$. In order to maximize the number of deliveries that can be completed by the fleet, each controller aims to minimize $M=2$ different objectives along the particular route: $V_1$ measures the time required to complete a particular delivery, and $V_2$ measures the battery charge required to complete a particular delivery.} 

\textit{Note that while true vehicle models are not linear, here we consider that the autonomous vehicle controller uses a linear approximation for path planning, to be tracked by a higher-fidelity reference tracking controller.}



For any given initial state $x_0 = x$, the feedback policy synthesis problem can be posed as the infinite-horizon, multi-objective optimal control problem 
\begin{align}\label{eq:inft_mo_opt}
V^{\star}(x) = \min_{\substack{\boldsymbol{\pi}}}&~~ (V_1(x_0, \boldsymbol{\pi}),~...~, V_M(x_0, \boldsymbol{\pi}))\\
    \text{s.t}&~~ x_{t+1}=Ax_t+B\pi_t(x_t)\nonumber\\
    &~~x_t\in\mathcal{X}, ~~\pi_t(x_t) \in \mathcal{U} ~~\forall t\geq 0\nonumber\\
    &~~x_0=x \nonumber \\
    = (V_1(&x, \boldsymbol{\pi}^{\star}), \dots,V_M(x, \boldsymbol{\pi}^{\star})) \nonumber
\end{align}
where the vector $V^{\star}(x) \in \mathbb{R}^{M}$ contains the $M$ objective values \eqref{vdef} associated with the optimal policy ${\boldsymbol{\pi}}^{\star}=\{\pi^{\star}_t(\cdot)\}_{t\geq0}$. The solution to \eqref{eq:inft_mo_opt} is formalized in Sec.~\ref{ssec:moo}.


There are two main challenges to solving the optimal control problem \eqref{eq:inft_mo_opt}.
First, \eqref{eq:inft_mo_opt} is infinite-dimensional, due to the unbounded time horizon and the optimization over functions $\boldsymbol{\pi} = \pi_0(\cdot),\pi_1(\cdot),..$ as decision variables. 
The aim of minimizing multiple objectives, which may be competing, introduces additional complexity in defining a solution to the problem. 
We propose a data-driven iterative approach to solving \eqref{eq:inft_mo_opt}. 
Towards addressing the infinite dimensionality of the optimal control problem, we adapt techniques from Learning MPC (LMPC) to iteratively approximate the optimal feedback policy using collected trajectory data of \eqref{eq:sysdyn_student}. 
To address the multi-objective nature of the optimal control problem, we use ideas prevalent in the multi-objective optimization literature for defining and obtaining a solution. 
We briefly review these concepts.

\subsection{Learning Model Predictive Control (LMPC)}

Consider the infinite-horizon optimal control problem as in \eqref{eq:inft_mo_opt}, but with a single control objective $V_1(x, \boldsymbol{\pi})$, defined as in \eqref{vdef}:
\begin{align}\label{eq:singleobjective}
   V^\star_1(x) = \min_{\substack{\boldsymbol{\pi}}}&~~ V_1(x_0, \boldsymbol{\pi})\\
    \text{s.t}&~~ x_{t+1}=Ax_t+B\pi_t(x_t)\nonumber\\
    &~~x_t\in\mathcal{X},~~\pi_t(x_t) \in \mathcal{U} ~~\forall t\geq 0\nonumber\\
    &~~x_0=x \nonumber
\end{align}
where $V^\star_1(x) = V_1(x,\boldsymbol{\pi}^\star)$ is the optimal control objective achievable from the initial state $x$, corresponding to an optimal sequence of feedback policies $\boldsymbol{\pi}^\star$. 

Learning Model Predictive Control (LMPC) \cite{rosolia2019learning} approximates a solution to \eqref{eq:singleobjective} by considering a system (\ref{eq:sysdyn_student}) performing the same optimal control task multiple times.
Let $x_t^j$, $u_t^j$ denote the state and input of the system respectively at time $t$ of iteration $j$, and with $\boldsymbol{\pi}^j= \{\pi^j_t(\cdot)\}_{t\geq0}$ the sequence of policies implemented during iteration $j$. 
We assume that at each iteration $j$ the system starts
from the same initial state, $x^j_0 = x_S,~ \forall j \geq 0$.
LMPC approximates the optimal policy of a single-objective, infinite-horizon problem like (\ref{eq:singleobjective}) by 
iteratively improving on an initial feasible sequence of policies, $\boldsymbol{\pi}^0 = \{\pi^0_t(\cdot)\}_{t\geq0}$. 

At iteration $j$, we define the convex set $\mathcal{CS}^j$ as
\begin{align}\label{eq:CS_def}
\mathcal{CS}^j=\text{conv}\left(\bigcup\limits_{r=0}^{j-1}\bigcup\limits_{t\geq0}\{x^r_t\}\right),
\end{align}
so that $\mathcal{CS}^j$ is the convex hull of all state trajectories from previous task iterations (see Assumption~\ref{assmp:cs0} on initialization procedures for $\mathcal{CS}^0$).
We define an estimate of the optimal control objective $V^{\star}_1(\cdot)$ on $\mathcal{CS}^j$ as
\begin{subequations}\label{eq:Tcost_def}
    \begin{align}
\hat{V}^{j,\star}_1(x)=\min_{\substack{\lambda^i_t\geq 0, ~\forall t\geq 0\\ \forall r=0,..,j-1}}& \sum_{r=0}^{j-1} \sum_{t \geq 0}\lambda^r_tV_1(x^r_{t},\boldsymbol{\pi}^r)\\\
    \text{s.t. }& \sum_{i=0}^{j-1} \sum_{t \geq 0} \lambda^r_t x^r_t=x,~~ \sum_{i=0}^{j-1} \sum_{t \geq 0} \lambda^r_t = 1.
\end{align}
\end{subequations}
By definition of $\mathcal{CS}^j$, each state $x \in \mathcal{CS}^j$ is a convex combination of states traversed during a previous iteration; for each state $x \in \mathcal{CS}^j$, $\hat{V}^{j,\star}_1(x)$ is the corresponding barycentric interpolation of the control objectives previously incurred from those states. 
Thus $\hat{V}^{j,\star}_1(\cdot)$ \eqref{eq:Tcost_def} is a data-driven, piecewise linear under-approximation of the convex value function $V^{\star}_1(\cdot)$ \eqref{eq:singleobjective}, estimated by interpolating the control objectives realized in previous trajectories. As more data is collected during additional iterations, the approximation error decreases.
As is convention for infeasible optimization problems, $\hat{V}^{j,\star}_1(x) = \infty$ for any $x \notin \mathcal{CS}^j$.

The LMPC algorithm uses \eqref{eq:CS_def} and \eqref{eq:Tcost_def} to approximate \eqref{eq:singleobjective} using MPC. 
At each time step $t$ of task iteration $j$, the LMPC algorithm solves the finite-horizon optimal control problem from the system's current state $x^j_t$:
\begin{equation}\label{eq:OG_LMPC}
	\begin{aligned}
 J^{\mathrm{LMPC},j}_{t}(x^j_t) =  \min\limits_{\mathbf{x}_t,\mathbf{u}_t } \quad & \displaystyle \hat{V}_1^{j,\star}(x_{t+N|t})+\sum\limits_{k=t}^{t+N-1} h_1(x_{k|t},u_{k|t}) \\
		\text{s.t.}\quad  & x_{k+1|t}=Ax_{k|t}+Bu_{k|t} \\
    & x_{k+1|t}\in\mathcal{X},~ u_{k|t} \in \mathcal{U}~~~~ \forall k \in [t, t+N-1]\\ 
    & x_{t|t} = x^j_t \\
    & x_{t+N|t}\in\mathcal{CS}^j,
	\end{aligned}
\end{equation}
where the decision variables $\mathbf{x}_t=[x_{t|t},..,x_{t+N|t}]$, $\mathbf{u}_t=[u_{t|t},..,u_{t+N-1|t}]$ are the state and input predictions along the prediction horizon $N$ made at time $t$. 
The optimal control problem \eqref{eq:OG_LMPC} searches for a feasible state and input trajectory of length $N$ beginning from the current state $x_t^j$ and ending in $\mathcal{CS}^j$. The objective function is the sum of the control stage cost $h_1(\cdot, \cdot)$ accumulated along the $N$-step trajectory plus the data-driven optimal control objective estimate $\hat{V}^{j,\star}_1(\cdot)$ evaluated the last predicted state.  
Let
\begin{subequations}\label{eq:soln}
\begin{align}
    \mathbf{u}^{j,\star}_t &= [u^{j, \star}_{t|t}, u^{j, \star}_{t+1|t},\dots, u^{j, \star}_{t+N-1|t}],~ 
    \mathbf{x}^{j,\star}_{t} = [x^{j, \star}_{t|t}, x^{j, \star}_{t+1|t},\dots, x^{j, \star}_{t+N|t}]
\end{align}    
\end{subequations}
be the optimal solution to \eqref{eq:OG_LMPC} at time $t$ of iteration $j$. 
Then \eqref{eq:soln} defines the time-invariant feedback policy as 
\begin{align}\label{eq:LMPC}
    u^j_t=\pi^{\mathrm{LMPC},j}(x^j_t)=u^{\star,j}_{t|t}.
\end{align}

At each time $t$, the LMPC controller \eqref{eq:OG_LMPC}-\eqref{eq:LMPC} plans a cost-optimal trajectory ending in $\mathcal{CS}^j$, the convex hull of states from which the system has successfully completed the task in a previous iteration. 
We refer to \cite{rosolia2019learning} for extensive analysis and modeling considerations for the single-objective linear LMPC formulation, including a formulation for jointly solving \eqref{eq:Tcost_def}-\eqref{eq:OG_LMPC} and proofs of recursive feasibility and iterative performance improvement, i.e. $V_1(x_0, \boldsymbol{\pi}^{\mathrm{LMPC},j}) \leq  V_1(x_0, \boldsymbol{\pi}^{\mathrm{LMPC},j-1})$.
In \cite{rosolia2022optimality}, authors show that the LMPC scheme converges asymptotically to the optimal cost of the infinite-horizon control problem \eqref{eq:singleobjective}.


\subsection{Multi-Objective Optimization}\label{ssec:moo}
Consider the generic multi-objective optimization problem
\begin{align}\label{eq:gen_mop}
    \min_{z\in\mathcal{Z}}~(f_1(z), ... , f_M(z))
\end{align}
where $z \in \mathbb{R}^{n_z}$ is the decision variable, $\mathcal{Z}$ the constraint set, and $f_0(\cdot), .., f_M(\cdot)$ distinct objective functions \cite[Chapter~4]{boyd2004convex}. 
A popular solution concept for multi-objective optimization problems is that of \textit{Pareto optimality}, which describes solutions for which no individual objective function $f_i(\cdot)$ can be decreased without causing an increase in another objective function.
\begin{definition} Given \eqref{eq:gen_mop}, a vector $z\in\mathcal{Z}$ is said to be Pareto optimal if $\nexists (z' \in \mathcal{Z}, i \in [1,M])$ such that $f_i(z')<f_i(z)$ and $f_j(z')\leq f_j(z)~\forall j \in [1,M] \setminus i$.
\end{definition}
Pareto optimal solutions to \eqref{eq:gen_mop} can be obtained by a scalarization procedure which converts \eqref{eq:gen_mop} into a single weighted-sum objective optimization problem. A common choice is to minimize a convex combination of objectives,
\begin{align}\label{eq:cvx_sclrzn}
    \min_{z\in\mathcal{Z}}~\sum_{i=1}^M \alpha_i f_i(z)
\end{align}
for some $\alpha$ such that $0 < \alpha_i < 1, ~\sum_{i=1}^M \alpha_i = 1$. When the objective functions $f_i(\cdot)$ and constraint set $\mathcal{Z}$ are convex, solving \eqref{eq:cvx_sclrzn} for a particular choice of $\alpha_i > 0$ produces a Pareto optimal solution $\bar{z}$, and conversely, \textit{every} Pareto optimal solution corresponds to a choice of $(\alpha_1,\dots, \alpha_M)$ \cite{miettinen1999nonlinear}.

\section{Multi-Objective LMPC}\label{sec:MOLMPC}

\begin{figure}[t]
    \centering
    \includegraphics[width=\columnwidth]{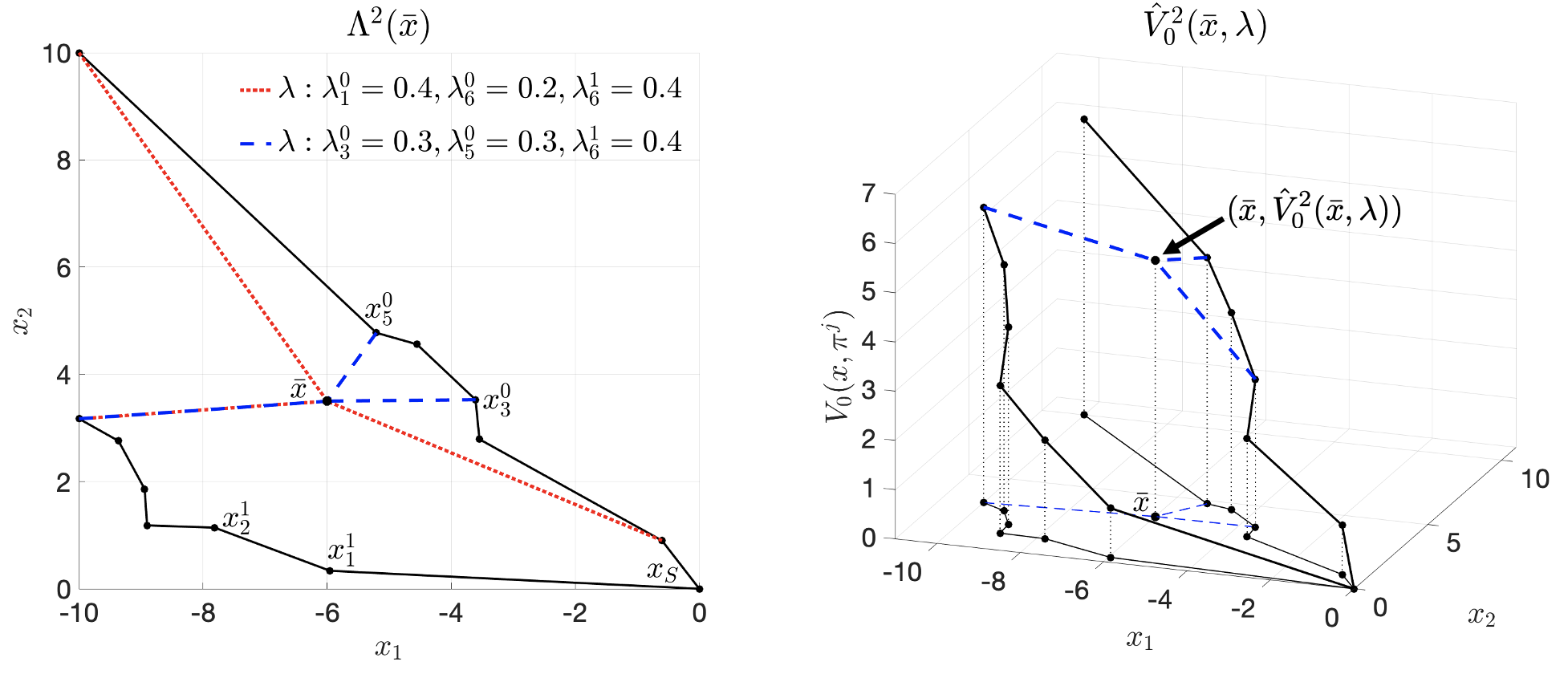} 
    \caption{The left plot depicts two possible values $\boldsymbol{\lambda} \in \Lambda^2(\bar{x})$ for a particular $\bar{x}$. Each $\boldsymbol{\lambda}$ corresponds to a particular convex combination of states from previous trajectories (here, segments of $\mathbf{x}^0$ and $\mathbf{x}^1$ are shown). The right plot depicts how $\hat{V}^{2,\star}_1(\bar{x},\boldsymbol{\lambda})$ is interpolated for a particular choice of $\boldsymbol{\lambda} \in \Lambda^2(\bar{x})$.} 
    \label{fig:lambda} 
\end{figure} 

Our approach adapts the LMPC formulation in \eqref{eq:OG_LMPC} to the multi-objective case, where our goal is to minimize multiple control costs over repeated task iterations. Specifically, we will formulate the Multi-Objective (MO-LMPC) to ensure that no single control objective $V_i$ \eqref{vdef} increases between any iterations $j\geq 0$ and $j+1$. We again assume that at each iteration $j$ the system starts from the same initial state, $x^j_0 = x_S,~ \forall j \geq 0$.

\textit{\textbf{Running Example:} 
Each time a delivery vehicle traverses a particular route between two delivery locations is considered an iteration $j$.
Note that the vehicles begin and end each of these iterations at the same set of geographical location $x_S$ and $x_F$. The high-level scheduler updates its estimated time and energy costs for each route based on recent vehicle performance. In order to ensure feasibility of the assigned schedule (i.e. no vehicle runs out of charge while completing a delivery or misses the assigned delivery window), we require that the low-level controllers never consume more time or energy on a route than during the previous iteration.}

Consider the convex safe set $\mathcal{CS}^j$ as defined at iteration $j$ in \eqref{eq:CS_def}.
By definition, each state $x\in\mathcal{CS}^j$ is associated with a set $\Lambda^j(x)$ of multipliers $\boldsymbol{\lambda}$, where
\begin{align}
   \Lambda^j(x)=\left\{ \boldsymbol{\lambda} = \bigcup_{r=0}^{j-1} \bigcup_{t\geq 0} \lambda^r_t ~~\middle\vert ~~\begin{aligned} &\sum_{i=0}^{j-1} \sum_{t \geq 0} \lambda^r_t x^r_t=x,\\
   &~ \sum_{i=0}^{j-1} \sum_{t \geq 0} \lambda^r_t = 1,~ \lambda^r_t\geq 0 \end{aligned} \right\}.
\end{align}
Each $\boldsymbol{\lambda}\in\Lambda^j(x)$ induces a specific set of $M$ control objective estimates $ \hat{V}^{j}_i$, one for each objective $V_i,~ i \in [1,M]$:
    \begin{align}\label{eq:Tcost_def_multi}
    \hat{V}^{j}_i(x, \boldsymbol{\lambda})= \begin{cases}
        \sum_{r=0}^{j-1} \sum_{t \geq 0}\lambda^r_tV_i(x^r_{t},\boldsymbol{\pi}^r) & \boldsymbol{\lambda} \in \Lambda^j(x)\\
        \infty & \text{else}.
    \end{cases}
\end{align}
The function $\hat{V}_i(x, \boldsymbol{\lambda})$ is a data-driven approximation of the value function $V^{\star}_i(x)$ \eqref{eq:singleobjective}, estimated by interpolating the single-objective value functions \eqref{vdef} associated with previous trajectories; in contrast with \eqref{eq:Tcost_def}, in \eqref{eq:Tcost_def_multi} the interpolation weights are specified by $\boldsymbol{\lambda}$. An example construction of $\hat{V}_i(x, \boldsymbol{\lambda})$ is shown in Fig.~\ref{fig:lambda}.
Note that since $V_i(\cdot, \cdot)$ is convex \eqref{vdef}, $\hat{V}_i(x, \boldsymbol{\lambda})$ \eqref{eq:Tcost_def_multi} is an upper bound of $V^{\star}_i(x)$ \eqref{eq:Tcost_def}.

For particular predicted state and input trajectories 
\begin{align*}
    &\mathbf{x}^j_t = [x^j_{t|t}, x^j_{t+1|t}, \dots, x^j_{t+N|t}],~\mathbf{u}^j_t = [u^j_{t|t}, u^j_{t+1|t}, \dots, u^j_{t+N-1|t}],
\end{align*}
and $\boldsymbol{\lambda}^j_t\in\Lambda^j(x^j_{t+N|t})$, we define for each $i \in [1, M]$ a running cost estimate:
\begin{align}\label{eq:finitehorizoncost}
    J^j_i(\mathbf{x}^j_t, \mathbf{u}^j_t, 
    \boldsymbol{\lambda}^j_t)=\hat{V}_i^{j-1}(x^j_{t+N|t}, \boldsymbol{\lambda}^j_t)&+\sum_{k=t}^{t+N-1} h_i(x^j_{k|t}, u^j_{k|t}).
\end{align}
Each $J^j_i(\mathbf{x}^j_t, \mathbf{u}^j_t, 
\boldsymbol{\lambda}^j_t)$ estimates the $i$th control objective accumulated from the current state $x^j_t$ to the end of the task. This estimate consists of the the $i$th stage cost accumulated over the predicted $N$-step trajectories $\mathbf{x}^j_t, \mathbf{u}^j_t$ plus the control objective estimate \eqref{eq:Tcost_def_multi} of the last predicted state in the $N$-step trajectory, evaluated using $\boldsymbol{\lambda}^j_t$. 

We now formulate the MO-LMPC problem for a particular choice  $\alpha \in \mathbb{R}^{M}$ where $0 < \alpha_i < 1, ~\sum_{i=1}^M \alpha_i = 1$  (we refer to Sec.~\ref{sec:example} for a discussion on the effect of choosing different $\alpha$). 
At each time $t$ of iteration $j \geq 1$, we solve:
\begin{subequations}\label{eq:OP_LMPC}
	\begin{align}
    J^{\mathrm{MO-LMPC},j}_{t,N}&(x_t^j) =  \nonumber\\
	\min\limits_{\mathbf{x}_t,\mathbf{u}_t, \boldsymbol{\lambda}_t } \quad & \displaystyle \sum\limits_{i=0}^{M}\alpha_iJ_i^j( \mathbf{x}_t, \mathbf{u}_t, \boldsymbol{\lambda}_t) \\
	\text{s.t.}\quad  & x_{k+1|t}=Ax_{k|t}+Bu_{k|t} \label{eq:dynamics}\\
    & x_{k+1|t}\in\mathcal{X},~ u_{k|t} \in \mathcal{U}~~~~~~~~~~~~~~~~~~~~~~~~ \forall k \in [t, t+N-1]\label{eq:stateconstraint}\\
    & x_{t+N|t}\in\mathcal{CS}^j \label{eq:convhullconstraint}\\
    & \boldsymbol{\lambda}_t\in\Lambda^j(x_{t+N|t}) \label{eq:termconstraint}\\
    & x_{t|t} = x^j_t \label{eq:initconstraint}\\
    &J^j_i(\mathbf{x}_t, \mathbf{u}_t, \boldsymbol{\lambda}_t)\leq \begin{cases}
        V_i(x_0, \boldsymbol{\pi}^{j-1}) & t = 0 \\
        J^j_i(\mathbf{x}^{j,\star}_{t-1}, \mathbf{u}^{j,\star}_{t-1}, \boldsymbol{\lambda}^{j,\star}_{t-1})-h_i(x^{j}_{t-1}, u^j_{t-1}) & \text{else}
    \end{cases} \nonumber\\
    & ~~~~~~~~~~~~~~~~~~~~~~~~~~~~~~~~~~~~~~~~~~~~~~~~~~~~~~~~~~~~ \forall i \in [1, M],\label{newconstraint}
	\end{align}
\end{subequations}
where $\mathbf{x}^{j,\star}_{t-1}, \mathbf{u}^{j,\star}_{t-1}, \boldsymbol{\lambda}^{j,\star}_{t-1}$ are the optimizers of $J_{t-1}^{\mathrm{MO-LMPC},j}(x_{t-1}^j)$ \eqref{eq:OP_LMPC} at the previous time step $t-1$. 
The controller then applies the first optimal input 
\begin{align}\label{eq:controlapplication}
u^j_t =\pi^{\mathrm{MO-LMPC}}(x_t) = u^\star_{t|t},
\end{align}
before re-solving \eqref{eq:OP_LMPC} at time $t+1$, resulting in a receding horizon control scheme.

As in single-objective LMPC \eqref{eq:OG_LMPC},
MO-LMPC \eqref{eq:OP_LMPC} searches for a feasible state trajectory that ends in the convex hull of all previous state trajectories.
MO-LMPC \eqref{eq:OP_LMPC} differs from single-objective LMPC \eqref{eq:OG_LMPC} in two main ways. 
First, the objective function in \eqref{eq:OP_LMPC} is now a sum of multiple control objectives. We have used the scalarization procedure outlined in \eqref{eq:cvx_sclrzn} to formulate the multi-objective optimization into a scalar optimization problem, using our choice of $\alpha$ to indicate potential prioritization between different objectives.
Second, we have incorporated a time-varying constraint \eqref{newconstraint}. Next we show this constraint \eqref{newconstraint} ensures that no objectives increase over successive iterations. 

Note \eqref{eq:finitehorizoncost} is a convex function in $\mathbf{x}_t$, $\mathbf{u}_t$, $\boldsymbol{\lambda}_t$, and the quantities $h_i(x^{j}_{t-1}, u^j_{t-1})$ and $J^j_i(\mathbf{x}^{j,\star}_{t-1}, \mathbf{u}^{j,\star}_{t-1}, \boldsymbol{\lambda}^{j,\star}_{t-1})$ are known scalars. 
Thus the MO-LMPC formulation \eqref{eq:OP_LMPC} maintains the convexity of the single-objective LMPC formulation \eqref{eq:OG_LMPC}.

\begin{remark}
\cite{vallon2024learning2} proposes an LMPC formulation which minimizes a single control objective while ensuring performance with respect to additional control objectives does not worsen. In contrast, the MO-LMPC formulation \eqref{eq:OP_LMPC} actively minimizes all control objectives. Results in Sec.~\ref{sec:example} demonstrate the differences between these methods.
Additionally, \eqref{eq:OP_LMPC} allows for a wider range of control objective formulations; control objectives may be any function of the state or input satisfying \eqref{eq:stagecost}.
\end{remark}

\begin{remark}
    The MO-LMPC method outlined here can straightforwardly be applied to systems modeled with nonlinear dynamics by using a sampled safe set $\mathcal{SS}^j$ defined as
$        \mathcal{SS}^j = \bigcup_{r=0}^{j-1} \bigcup_{t\geq 0}\{x^r_t\}$
    in place of the convex set $\mathcal{CS}^j$ or exploiting system theoretic properties (e.g. feedback linearizability \cite{nair2021output}, monotonicity \cite{rosolia2021minimum})  to compute the convex safe set $\mathcal{CS}^j$. The properties of stability and multi-objective performance improvement demonstrated in the following section hold for this nonlinear formulation, though the problem \eqref{eq:OP_LMPC} is no longer convex. 
\end{remark}

\begin{assumption}\label{assmp:cs0}
    At iteration $j=1$, we assume the set $\mathcal{CS}^{j-1} = \mathcal{CS}^0$ contains a feasible state trajectory (with respect to \eqref{eq:inft_mo_opt}) converging to $x_F$ \eqref{eq:xf}.     
\end{assumption}
Assumption~\ref{assmp:cs0} is not restrictive in practice: any initial suboptimal trajectory satisfying the constraints of \eqref{eq:inft_mo_opt} and converging to $x_F$ can initialize $\mathcal{CS}^0$. 

\textit{\textbf{Running Example:} 
In our example, the set $\mathcal{CS}^0$ will contain trajectories between two delivery locations that result in conservative time and energy usages. This might correspond to a controller driving strictly below the speed limit.}

\textit{Note that such a conservative initialization may mean the high-level scheduler can not confidently plan as many delivery routes as desired. As more data is collected and the low-level controller performance improves with each iteration, the time and energy cost estimates along each route are lowered. This will allow the high-level scheduler to assign increasingly productive delivery routes to vehicles.}

\subsection{ Properties}\label{sec:properties}

\begin{theorem}[Feasibility and Stability]\label{thm:feas}
Consider system \eqref{eq:sysdyn_student} controlled by the MO-LMPC
controller \eqref{eq:OP_LMPC} - \eqref{eq:controlapplication}. Let $\mathcal{CS}^j$
be the sampled safe set at
iteration $j$ as defined in \eqref{eq:CS_def}. Let Assumption~\ref{assmp:cs0} hold. Then,
the MO-LMPC \eqref{eq:OP_LMPC} - \eqref{eq:controlapplication} is feasible for all time steps $t \geq 0$ at
every iteration $j \geq 1$. Moreover, the equilibrium point $x_F$ \eqref{eq:xf}
is asymptotically stable for the closed loop system \eqref{eq:sysdyn_student}, \eqref{eq:OP_LMPC} - \eqref{eq:controlapplication} at every iteration $j\geq 1$.
\end{theorem}
\begin{proof}
    By Assumption \ref{assmp:cs0}, the set $\mathcal{CS}^0$ is non-empty. Since $\mathcal{CS}^{j-1}\subseteq \mathcal{CS}^j$, it follows that $\mathcal{CS}^j$ is also a non-empty set. Furthermore, we have $x_0^j = x_S,~ \forall j \geq 0$.

    At time $t=0$ of iteration $j \geq 1$, consider the $N$-step state trajectory, corresponding input sequence and convex multiplier vector given by
    \begin{subequations}\label{eq:candidate_traj}
          \begin{align}
        &\mathbf{x}_0' = [x_0^{j-1}, x_1^{j-1}, \dots, x_N^{j-1}],~\mathbf{u}_0' = [u_0^{j-1}, u_1^{j-1}, \dots, u_{N-1}^{j-1}]\\
        & \boldsymbol{\lambda}'_0 = \bigg\{ \bigcup_{r=0}^{j-1} \bigcup_{t \geq 0} \lambda_t^r ~\vert ~ \lambda_N^{j-1}=1,~ \lambda_t^r = 0 ~ \text{else} \bigg\}.
    \end{align}  
    \end{subequations}
    Note \eqref{eq:candidate_traj} correspond to the state and input trajectories taken by the system in the previous iteration $j-1$, and therefore
    satisfy constraints \eqref{eq:dynamics}, \eqref{eq:stateconstraint}, and \eqref{eq:termconstraint}. 
    At time $t=0$, it remains to verify if \eqref{eq:candidate_traj} satisfies constraint \eqref{newconstraint}. From \eqref{eq:finitehorizoncost}, \eqref{eq:Tcost_def_multi} we have for all $i \in [1,M]$:
    \begin{align*}
        J^j_i(\mathbf{x}_0', \mathbf{u}_0', \boldsymbol{\lambda}'_0)
        &= \sum_{t=0}^{N-1}h_i(x_t^{j-1}, u_t^{j-1}) + \hat{V}_i^{j-1}(x_N^{j-1}, \boldsymbol{\lambda}'_0) \\
        &= \sum_{t=0}^{N-1}h_i(x_t^{j-1}, u_t^{j-1}) + \sum_{r=0}^{j-1} \sum_{t\geq 0 } \lambda_t^r V_i(x_t^r, \boldsymbol{\pi}^r) \\
        &= \sum_{t=0}^{N-1}h_i(x_t^{j-1}, u_t^{j-1}) + V_i(x_N^{j-1}, \boldsymbol{\pi}^{j-1}) \\
        &= \sum_{t=0}^{\infty} h_i(x_t^{j-1}, u_t^{j-1}) = V_i(x_S, \boldsymbol{\pi}^{j-1}).
    \end{align*}
    Thus, \eqref{eq:candidate_traj} satisfies \eqref{newconstraint} with equality, and \eqref{eq:candidate_traj} is a feasible solution to \eqref{eq:OP_LMPC} at time $t=0$ of iteration $j\geq1$. 

    Assume the MO-LMPC \eqref{eq:OP_LMPC}-\eqref{newconstraint} is feasible at time $t$ of iteration $j$, with 
\begin{subequations}\label{eq:candidate_traj_2}
          \begin{align}
        &\mathbf{x}^{j,\star}_t = [x_{t|t}^{j,\star}, x_{t+1|t}^{j, \star}, \dots, x_{t+N|t}^{j,\star}],~\mathbf{u}^{j,\star}_t = [u_{t|t}^{j,\star}, u_{t+1|t}^{j, \star}, \dots, u_{t+N-1|t}^{j,\star}] \\
        & \boldsymbol{\lambda}^{j,\star}_t = \bigg\{ \bigcup_{r=0}^{j-1} \bigcup_{k \geq 0} \lambda_{k|t,j}^{r,\star} \bigg\}.
    \end{align}  
    \end{subequations}
    Because of \eqref{eq:controlapplication}, \eqref{eq:initconstraint}, and no model mismatch between \eqref{eq:sysdyn_student}, \eqref{eq:stateconstraint} we have
    $u_t^j = u^{j,\star}_{t|t},~ x_t^j = x^{j,\star}_{t|t},~ x^j_{t+1} = x^{j,\star}_{t+1|t}.$
    It further follows from \eqref{eq:convhullconstraint} that
        $x^{j, \star}_{t+N|t} = \sum_{r = 0}^{j-1} \sum_{k\geq 0} \lambda^{r, \star}_{k|t,j} x^r_k.$
    Let us define 
    \begin{align*}
        \bar{u} = \sum_{r=0}^{j-1} \sum_{k \geq 0 } \lambda^{r, \star}_{k|t,j} u^r_k,~~~ \bar{x} &= \sum_{r=0}^{j-1} \sum_{k \geq 0 } \lambda^{r, \star}_{k|t,j} x^r_{k+1} = \sum_{r=0}^{j-1} \sum_{k \geq 1} \lambda^{r, \star}_{k-1|t,j} x^r_{k},
   \end{align*}
    where $\bar{u} \in \mathcal{U}$, and $\bar{x} \in \mathcal{CS}^{j-1}$.
    Consider the candidate trajectories at time $t+1$,
\begin{align}\label{eq:beepboop}
        &\mathbf{x}_{t+1}' = [x_{t+1|t}^{j, \star}, x_{t+2|t}^{j, \star}, \dots, x_{t+N|t}^{j, \star}, \bar{x}], ~\mathbf{u}_{t+1}' = [u_{t+1|t}^{j, \star}, u_{t+2|t}^{j, \star}, \dots, u_{t+N-1|t}^{j, \star}, \bar{u}]\nonumber \\
        & \boldsymbol{\lambda}'_{t+1} = \{\lambda~|~ \lambda^r_k = \lambda^{r,\star}_{k-1|t,j}, \lambda^r_0 = 0 \}
    \end{align} 
    which satisfy constraints \eqref{eq:dynamics}, \eqref{eq:stateconstraint}, and \eqref{eq:termconstraint}. 
    At time $t+1$, it remains to verify if \eqref{eq:beepboop} satisfies constraint \eqref{newconstraint}.
    From \eqref{eq:finitehorizoncost}, \eqref{eq:Tcost_def_multi}, we have for all $i \in [1, M]$:    
    \begingroup
    \allowdisplaybreaks
\begin{subequations}\label{eq:aa}
        \begin{align}
            J^j_i(&\mathbf{x}_{t+1}',  \mathbf{u}_{t+1}',\boldsymbol{\lambda}'_{t+1}) = \sum_{k=t+1}^{t+N} h_i(x^{j, \star}_{k|t}, u^{j, \star}_{k|t}) + \hat{V}^{j-1}_i(\bar{x}, \boldsymbol{\lambda}'_{t+1})  \\
            &= J^j_i(\mathbf{x}^j_t, \mathbf{u}^j_t, \mathbf{\lambda}^j_t) - h_i(x^{j, \star}_{t|t},u^{j, \star}_{t|t}) + \hat{V}^{j-1}_i(\bar{x}, \boldsymbol{\lambda}'_{t+1}) - \hat{V}^{j-1}_i(x^{j, \star}_{t+N|t}, \boldsymbol{\lambda}^j_t) ~+\nonumber\\
            &~~~~~  + h_i(x^{j, \star}_{t+N|t}, \bar{u}) \\
            &= J^j_i(\mathbf{x}^j_t, \mathbf{u}^j_t, \mathbf{\lambda}^j_t) - h_i(x^{j}_{t},u^{j}_{t}) + \hat{V}^{j-1}_i(\bar{x}, \boldsymbol{\lambda}'_{t+1}) -  \sum_{r = 0}^{j-1} \sum_{k \geq 0} \lambda^{r, \star}_{k | t,j} V_i(x_k^r, \boldsymbol{\pi}^r)~+\nonumber \\
            & ~~~~~+ h_i(x^{j, \star}_{t+N|t}, \bar{u})  \\
            &= J^j_i(\mathbf{x}^j_t, \mathbf{u}^j_t, \mathbf{\lambda}^j_t) - h_i(x^{j}_{t},u^{j}_{t})  + \hat{V}^{j-1}_i(\bar{x}, \boldsymbol{\lambda}'_{t+1}) - 
            \sum_{r = 0}^{j-1} \sum_{k \geq 0} \lambda^{r, \star}_{k | t,j} V_i(x^r_{k+1}, \boldsymbol{\pi}^r) ~+ \nonumber \\
            & ~~~~~ + h_i(x^{j, \star}_{t+N|t}, \bar{u}) -  \sum_{r = 0}^{j-1} \sum_{k \geq 0} \lambda^{r, \star}_{k | t,j} h_i(x_k^r, u_k^r).
        \end{align}
    \end{subequations}
    \endgroup
    From Jensen's inequality, it follows that
    \begin{align}\label{eq:bb}
        h_i(x^{j, \star}_{t+N|t}, \bar{u}) \leq \sum_{r = 0}^{j-1} \sum_{k \geq 0} \lambda^{r, \star}_{k | t,j} h_i(x_k^r, u_k^r).
    \end{align}
    Furthermore, by definition of $\boldsymbol{\lambda}'_{t+1}$ in \eqref{eq:beepboop}, we have
    \begin{align}\label{eq:cc}
        \hat{V}^{j-1}_i(\bar{x}, \boldsymbol{\lambda}'_{t+1}) &= \sum_{r = 0}^{j-1} \sum_{k \geq 1} \lambda^{r,\star}_{k|t,j} V_i(x_k^r, \boldsymbol{\pi}^r) = \sum_{r = 0}^{j-1} \sum_{k \geq 0} \lambda^{r, \star}_{k | t,j} V_i(x^r_{k+1}, \boldsymbol{\pi}^r).
    \end{align}
    From \eqref{eq:aa}, \eqref{eq:bb}, \eqref{eq:cc}, it follows that
    \begin{align*}
        J_i^j(\mathbf{x}_{t+1}', \mathbf{u}_{t+1}',\boldsymbol{\lambda}'_{t+1}) \leq J^j_i(\mathbf{x}^j_t, \mathbf{u}^j_t, \mathbf{\lambda}^j_t) - h_i(x^{j}_{t},u^{j}_{t})~~~ \forall i \in [1,M],
    \end{align*}
    and thus \eqref{eq:beepboop} satisfies \eqref{newconstraint} and is a feasible solution to \eqref{eq:OP_LMPC} at time $t+1$ of iteration $j$. We have shown that at iteration $j \geq 1$, the MO-LMPC is feasible at time $t=0$, and that feasibility of the MO-LMPC at time $t$ implies feasibility at time $t+1$. We conclude by induction that the MO-LMPC \eqref{eq:OP_LMPC}-\eqref{eq:controlapplication} is feasible at all time $t \geq 0$ of all iterations $j \geq 1$.

    Next we show asymptotic stability of $x_F$ \eqref{eq:xf} under the MO-LMPC policy \eqref{eq:OP_LMPC}-\eqref{eq:controlapplication} by showing that the cost function $J^{\mathrm{MO-LMPC},j}_{t,N}(\cdot)$ \eqref{eq:OP_LMPC} is a time-varying Lyapunov function for the closed-loop system \eqref{eq:sysdyn_student}, \eqref{eq:controlapplication} for the equilibrium state $x_F$ \eqref{eq:xf}. 
    Continuity of $J^{\mathrm{MO-LMPC},j}_{t,N}(\cdot)$ can be shown from \cite{dwayne}, and
    it follows from \eqref{eq:stagecost} that 
    \begin{align}
        &J^{\mathrm{MO-LMPC},j}_{t,N}(x) \succ 0 ~\forall x \in \mathbb{R}^{n_x} \backslash \{x_F\},~ J^{\mathrm{MO-LMPC},j}_{t,N}(x_F) = 0.
    \end{align}
    Moreover, recursive feasibility of \eqref{eq:OP_LMPC} and compactness of the constraints imply that $J^{\mathrm{MO-LMPC},j}_{t,N}(x)$ can be bounded above by some positive definite function $W(x)$ independent of time $t$.
    It remains to show that $J^{\mathrm{MO-LMPC}, j}_{t,N}(\cdot)$ decreases along the closed-loop trajectory, where
    \begin{align}
    J^{\mathrm{MO-LMPC}, j}_{t,N}(x^j_t) = \sum_{i=0}^{M-1} \alpha_i J^j_i(\mathbf{x}^{j,\star}_t, \mathbf{u}^{j,\star}_t, \boldsymbol{\lambda}^{j, \star}_t).
    \end{align}
    Note constraint \eqref{newconstraint} ensures that for all $i \in [1,M]$
    \begin{align}
        J^j_i(\mathbf{x}^{j,\star}_t, \mathbf{u}^{j,\star}_t, \boldsymbol{\lambda}^{j, \star}_t) \leq  J^j_i(\mathbf{x}^{j,\star}_{t-1}, \mathbf{u}^{j,\star}_{t-1}, \boldsymbol{\lambda}^{j, \star}_{t-1}) - h_i(x^j_{t-1}, u^j_{t-1}), \nonumber 
    \end{align}
    and since all $\alpha_i \geq 0$, it follows that
    \begin{align}
        \sum_{i=1}^M \alpha_i J^j_i(\mathbf{x}^{j,\star}_t, \mathbf{u}^{j,\star}_t, \boldsymbol{\lambda}^{j, \star}_t) &\leq  \sum_{i=1}^M \alpha_i J^j_i(\mathbf{x}^{j,\star}_{t-1}, \mathbf{u}^{j,\star}_{t-1}, \boldsymbol{\lambda}^{j, \star}_{t-1}) - \sum_{i=1}^M \alpha_i h_i(x^j_{t-1}, u^j_{t-1}), \nonumber
    \end{align}
    and therefore
    \begin{align}\label{eq:stabilityresult}
        J^{\mathrm{MO-LMPC}, j}_{t,N}(x_t^j) -  J^{\mathrm{MO-LMPC}, j}_{t-1,N}(x_{t-1}^j) & \leq - \sum_{i=1}^M \alpha_i h_i(x^j_{t-1}, u^j_{t-1}) < 0.
    \end{align}
    
    Thus, by \eqref{eq:stabilityresult}, positive definiteness of each stage cost $h_i(\cdot, \cdot)$ \eqref{eq:stagecost}, and continuity of $J^{\mathrm{MO-LMPC}, j}_{t,N}{\cdot}$ \eqref{eq:finitehorizoncost}, we conclude that the state $x_F$ \eqref{eq:xf} is asymptotically stable for the system \eqref{eq:sysdyn_student} in closed-loop with the MO-LMPC \eqref{eq:OP_LMPC}-\eqref{eq:controlapplication}\hfill $\blacksquare$
\end{proof}

\begin{theorem}[Performance Improvement]\label{thm:cost}
   Consider system \eqref{eq:sysdyn_student} in closed-loop with the MO-LMPC controller \eqref{eq:OP_LMPC}-\eqref{eq:controlapplication}. Let $\mathcal{CS}^{j}$ be the sampled safe set at the iteration $j$ as defined in \eqref{eq:CS_def}. Let Assumption \eqref{assmp:cs0} hold. 
   Then, no individual control objective $V_i$ increases over successive iterations $j$ and $j-1$.
   i.e. $V_i(x_S, \boldsymbol{\pi}^j) \leq V_i(x_S, \boldsymbol{\pi}^{j-1}),~ \forall i \in [1,M].$
\end{theorem}
\begin{proof}
    At time $t=0$ of iteration $j$, the constraint \eqref{newconstraint} enforces $\forall i \in [1,M]$
    \begin{align}
        J^j_i(\mathbf{x}^{j,\star}_0, \mathbf{u}^{j,\star}_0, \boldsymbol{\lambda}^{j,\star}_0) &\leq V_i (x_0^j, \boldsymbol{\pi}^{j-1}) =\sum_{t=0}^{\infty} h_i(x_t^{j-1}, u_t^{j-1}).\label{eq:upper-bound}
    \end{align}
    Similarly, the constraints at time $t \geq 1$ enforce 
    \begin{align}
        J^j_i(\mathbf{x}^{j,\star}_0, \mathbf{u}^{j,\star}_0, \boldsymbol{\lambda}^{j,\star}_0) &\geq J_i^j(\mathbf{x}^{j,\star}_1, \mathbf{u}^{j,\star}_1, \boldsymbol{\lambda}^{j,\star}_1) + h_i(x^j_0, u^j_0) \nonumber\\
        & \geq  J_i^j(\mathbf{x}^{j,\star}_2, \mathbf{u}^{j,\star}_2, \boldsymbol{\lambda}^{j,\star}_2) + \sum_{t=0}^1 h_i(x_t^j, u_t^j) \nonumber\\
        & \geq \lim_{k \rightarrow \infty} \big[J_i^j(\mathbf{x}_k^{j,\star}, \mathbf{u}_k^{j,\star}, \boldsymbol{\lambda}^{j,\star}_k)~ + \sum_{t=0}^{k-1}h_i(x_t^j, u_t^j)\big].
    \end{align}
    From Thm.~\ref{thm:feas}, we have $\lim_{t \rightarrow \infty} x_t^j = x_F,$
    and therefore by continuity of each $h_i(\cdot, \cdot)$, $\lim_{t \rightarrow \infty } J_i^j(\mathbf{x}_t^{j,\star}, \mathbf{u}_t^{j,\star}, \boldsymbol{\lambda}^{j,\star}_t) = 0,$
    resulting in 
\begin{align}\label{eq:lowerbound}
        J_i^j(\mathbf{x}_0^j, \mathbf{u}_0^j, \boldsymbol{\lambda}^j_0) \geq \sum_{t=0}^{\infty} h_i(x_t^j, u_t^j).
    \end{align}
    From \eqref{eq:upper-bound} and \eqref{eq:lowerbound} it follows that 
    \begin{align}\label{eq:individual-cost-result}
        \sum_{t=0}^{\infty} h_i(x_t^{j-1}, u_t^{j-1}) \geq J_i^j(\mathbf{x}_0^j, \mathbf{u}_0^j, \boldsymbol{\lambda}^j_0) \geq \sum_{t=0}^{\infty} h_i(x_t^{j}, u_t^{j}),
    \end{align}
    from which we conclude that each objective $V_i$ exhibits non-increase between successive iterations. 
   \hfill$\blacksquare$
    
\end{proof}

Next, we consider the convergence properties of the MO-LMPC controller \eqref{eq:OP_LMPC}-\eqref{eq:controlapplication}. We will make use of the following Lemma, which compares the optimal solutions between the optimal control problem
\begin{subequations}\label{eq:constrained}
    \begin{align}
    \min_{x \in \mathbb{R}^{n_x}} & \sum_{i=1}^M \alpha_i f_i(x) \\
    \text{s.t. }& Ax = b \\
    & Hx \leq g \\
    & f_i(x) \leq r_i ~~ \forall i \in [1,M]
\end{align}
\end{subequations}
where $A \in \mathbb{R}^{n_b \times n_x}$, $b \in \mathbb{R}^{n_b}$, $H \in \mathbb{R}^{n_g \times n_x}$, and $g \in \mathbb{R}^{n_g}$,
and the variation
\begin{subequations}\label{eq:unconstrained}
    \begin{align}
    \min_{x \in \mathbb{R}^{n_x}} & \sum_{i=1}^M \bar{\alpha}_i f_i(x) \\
    \text{s.t. }& Ax = b \\
    & Hx \leq g
\end{align}
\end{subequations}
where the scalarization coefficients $\alpha_i > 0$ and $\bar{\alpha}_i > 0$ weight different convex objectives $f_i(x) ~ i \in [1,M]$. 

\begin{lemma}\label{lem:equiv_mo_prob}
    Let Slater's condition hold for both the problems \eqref{eq:constrained} and \eqref{eq:unconstrained}. Assume \eqref{eq:constrained} is feasible, with optimal solution $x^{\star}$. Then, $x^{\star}$ is also the optimal solution to \eqref{eq:unconstrained} for a particular choice of $\bar{\alpha} > 0$. 
\end{lemma}
\begin{proof}
    Consider the convex optimization problem \eqref{eq:constrained}. For any optimal solution $x^{\star}$, by strong duality (from Slater's condition) there exist corresponding KKT multipliers $\lambda^{\star} \in \mathbb{R}^{n_g + M + 1}$ and $\nu^{\star} \in \mathbb{R}^{n_b}$ jointly satisfying
    \begin{subequations}\label{eq:constrainedKKT}
        \begin{align}
        &Ax^{\star}= b \\
        &Hx^{\star} \leq g \\
        &f_i(x^{\star}) \leq r_i ~~~ \forall i \in [1,M]\\
        &\lambda^{\star} \geq 0\\
        &\lambda^{\star}_i(f_i(x^{\star}) - r_i) = 0 ~~~ \forall i \in [1,M] \\
        &\lambda^{\star\top}_{M : M+n_g}(Hx^{\star}-g)=0\\
        &\sum_{i=0}^{M}\alpha_i \nabla f_i(x^{\star}) + \sum_{i=1}^M \lambda_i^{\star} \nabla f_i(x^{\star}) + H^{\top} \lambda^{\star}_{M:M+n_g} + A^{\top}\nu^{\star} = 0.
    \end{align}
    \end{subequations}
    Similarly, we know that any optimal solution $\bar{x}^{\star}$ to \eqref{eq:unconstrained} must be associated with corresponding KKT multipliers $\bar{\lambda}^{\star} \in \mathbb{R}^{M}$ and $\bar{\nu}^{\star} \in \mathbb{R}^{n_g}$ that jointly satisfy    \begin{subequations}\label{eq:unconstrainedKKT}
        \begin{align}
        &A\bar{x}^{\star}= b \\
        &H\bar{x}^{\star} \leq g \\
        &\bar{\lambda}^{\star} \geq 0\\
        &\bar{\lambda}^{\star\top}(H\bar{x}^{\star}-g)=0\\
        &\sum_{i=0}^{M}\bar{\alpha}_i \nabla f_i(\bar{x}^{\star}) + H^{\top} \bar{\lambda}^{\star} + A^{\top}\bar{\nu}^{\star} = 0.
    \end{align}
    \end{subequations}
    Now, assume 
    \begin{align}\label{eq:baralphaidef}
        \bar{\alpha}_i = \alpha_i + \lambda^{\star}_i
    \end{align}
    Then, 
    $
        \bar{x}^{\star} = x^{\star},~ \bar{\lambda}^{\star} = \lambda^{\star}_{0:M},~~ \bar{\nu}^{\star} = \nu^{\star}
    $
    satisfy \eqref{eq:unconstrainedKKT}. 
    Note that since $\lambda^{\star}\geq 0$ and $\alpha > 0$, it follows that $\bar{\alpha}_i > 0$ when defined according to \eqref{eq:baralphaidef}. 
    Thus we conclude that the optimal solution $x^{\star}$ to \eqref{eq:constrained} is also the optimal solution $\bar{x}^{\star}$ to \eqref{eq:unconstrained} for a particular choice of multipliers $\bar{\alpha}_i >0$. \hfill
    $\blacksquare$
\end{proof}

\begin{assumption}\label{assmp:strng_cvx}
    There exists at least one $i'\in[1, M]$ such that $h_{i'}(\cdot, \cdot)$ is strongly convex and $\alpha_{i'}>0$.
\end{assumption}
\begin{theorem}
    [Pareto Optimality]\label{thm:pareti}
   Consider system \eqref{eq:sysdyn_student} in closed-loop with the MO-LMPC controller \eqref{eq:OP_LMPC}-\eqref{eq:controlapplication}. Let $\mathcal{CS}^{j}$ be the sampled safe set at the iteration $j$ as defined in \eqref{eq:CS_def}. Let Assumptions \eqref{assmp:cs0} and \eqref{assmp:strng_cvx} hold. 
   Assume that after a finite number of iterations $c$, the closed-loop trajectory converges to fixed-point state and input trajectories such that $(\mathbf{x}^j, \mathbf{u}^j)=(\mathbf{x}^\infty, \mathbf{u}^\infty)$, $\forall j\geq c$. Then the policy $\boldsymbol{\pi}^j$ inducing the closed-loop trajectory $(\mathbf{x}^\infty, \mathbf{u}^\infty)$ is a Pareto optimal solution to the infinite-horizon multi-objective optimal control problem \eqref{eq:inft_mo_opt}.
\end{theorem}

\begin{proof}
    First, we show that $x^\infty_t,\dots, x^\infty_{t+N}$, $u^\infty_t,\dots, u^\infty_{t+N-1}$ is the optimizer of \eqref{eq:OP_LMPC} for any $N$, $\forall t\geq 0$ and $\forall j\geq c+1$. Then we show that the closed-loop trajectory $(\mathbf{x}^\infty, \mathbf{u}^\infty)$ can be equivalently given by the solution to \eqref{eq:OP_LMPC} at $t=0$ as $N\rightarrow\infty$. Then we show that \eqref{eq:OP_LMPC} as $N\rightarrow\infty$ is equivalent to the scalarization of \eqref{eq:inft_mo_opt}, with $(\mathbf{x}^\infty, \mathbf{u}^\infty)$ as its optimizer.

    For any $j\geq c+1$, we have from \eqref{eq:lowerbound} and (31) that at $x^j_0 = x^\infty_0$,
    \begin{subequations}
        \begin{align}
        J^{\mathrm{MO-LMPC},j}_{0,N}(x^j_0) &\geq \sum_{i=1}^M \alpha_i \sum_{k=0}^{\infty} h_i(x_k^{j}, u_k^{j})\\
        &= \sum_{i=1}^M \alpha_i \sum_{k=t}^{N-1} h_i(x_k^{j}, u_k^{j}) + \sum_{i=1}^M \alpha_i \sum_{k=t+N}^{\infty} h_i(x_k^{j}, u_k^{j})\\
        &= \sum_{i=1}^M \alpha_i \sum_{k=0}^{N-1} h_i(x_k^{c}, u_k^{c}) + \sum_{i=1}^M \alpha_i \sum_{k=N}^{\infty} h_i(x_k^{c}, u_k^{c})\\
        & = \sum_{i=1}^M \alpha_i \sum_{k=0}^{N-1} h_i(x_k^{c}, u_k^{c}) +\sum_{i=1}^M \alpha_i \hat{V}_i^{j}(x_{N}^{c}, \bar{\boldsymbol{\lambda}}_0)
    \end{align}
    \end{subequations}
    where $\bar{\boldsymbol{\lambda}}_0\in\Lambda^j(x^j_{N})$ is such that $\lambda^c_{N}=1$. The second equality is obtained because $(\mathbf{x}^j, \mathbf{u}^j) = (\mathbf{x}^\infty, \mathbf{u}^\infty) $, $\forall j\geq c+1$ and the last equality is obtained using the definition \eqref{eq:Tcost_def_multi}. Notice that because $(\mathbf{x}^j, \mathbf{u}^j) = (\mathbf{x}^\infty, \mathbf{u}^\infty) $, we have that the state-input trajectory $\bar{\mathbf{x}}^\infty_{0:N} =[x^\infty_0,\dots, x^\infty_{N}]$, $\bar{\mathbf{u}}^\infty_{0:N}=[u^\infty_0,\dots, u^\infty_{N-1}]$, $\bar{\boldsymbol{\lambda}}_0$ is a feasible solution to \eqref{eq:OP_LMPC} at time $t=0$, and also provides a lower bound for $J^{\mathrm{MO-LMPC},j}_{0,N}(x^j_0)$. Hence, it must be optimal. Now suppose $\bar{\mathbf{x}}^\infty_{t:t+N} =[x^\infty_t,\dots, x^\infty_{t+N}]$, $\bar{\mathbf{u}}^\infty_{t:t+N-1} =[u^\infty_t,\dots, u^\infty_{t+N-1}]$ and multipliers $\bar{\boldsymbol{\lambda}}_t$ such that $\lambda^c_{t+N}=1$, is optimal for \eqref{eq:OP_LMPC} at time $t$. We show that then $\bar{\mathbf{x}}^\infty_{t+1:t+N+1} =[x^\infty_{t+1},\dots, x^\infty_{t+N+1}]$, $\bar{\mathbf{u}}^\infty_{t+1:t+N} =[u^\infty_{t+1},\dots, u^\infty_{t+N}]$ and multipliers $\bar{\boldsymbol{\lambda}}_{t+1}$ (defined analogously to $\bar{\boldsymbol{\lambda}}_t$) are optimal at time $t+1$. As before, we can show that 
    \begin{align}
        J^{\mathrm{MO-LMPC},j}_{t+1,N}(x^j_{t+1}) \geq  \sum_{i=1}^M \alpha_i\left( \sum_{k=t+1}^{t+N} h_i(x_k^{c}, u_k^{c}) +\hat{V}_i^{j}(x_{t+N+1}^{c}, \bar{\boldsymbol{\lambda}}_{t+1})\right)
    \end{align}
to obtain a lower bound on $J^{\mathrm{MO-LMPC},j}_{t+1,N}(x^j_{t+1})$. It remains to show that $\bar{\mathbf{x}}^\infty_{t+1:t+N+1}$, $\bar{\mathbf{u}}^\infty_{t+1:t+N}$, $\bar{\boldsymbol{\lambda}}_{t+1}$ is feasible and satisfies \eqref{newconstraint}. We verify that
\begin{subequations}
    \begin{align}
    J^j_i(\bar{\mathbf{x}}^\infty_{t+1:t+N+1}, \bar{\mathbf{u}}^\infty_{t+1:t+N},&\bar{\boldsymbol{\lambda}}_{t+1})+h_i(x^\infty_t, u^\infty_t)=\\
    &=\sum_{k=t}^{t+N}h_i(x^\infty_k, u^\infty_k) + \hat{V}^j_i(x^\infty_{t+N+1}, \bar{\boldsymbol{\lambda}}_{t+1} )\\
    &=\sum_{k=t}^{t+N}h_i(x^\infty_k, u^\infty_k)+\sum_{k=t+N+1}^{\infty}h_i(x^\infty_k, u^\infty_k)\\
    &=\sum_{k=t}^{t+N-1}h_i(x^\infty_k, u^\infty_k)+\sum_{k=t+N}^{\infty}h_i(x^\infty_k, u^\infty_k)\\
    &=\sum_{k=t}^{t+N-1}h_i(x^\infty_k, u^\infty_k)+\hat{V}^j_i(x^\infty_{t+N}, \bar{\boldsymbol{\lambda}}_{t} )\\
    & = J_i(\bar{\mathbf{x}}^\infty_{t:t+N}, \bar{\mathbf{u}}^\infty_{t:t+N-1}, \bar{\boldsymbol{\lambda}}_{t}).
\end{align}
\end{subequations}
Thus, $\bar{\mathbf{x}}^\infty_{t+1:t+N+1}, \bar{\mathbf{u}}^\infty_{t+1:t+N}, \bar{\boldsymbol{\lambda}}_{t+1}$ is optimal for \eqref{eq:OP_LMPC} at time $t+1$ if $\bar{\mathbf{x}}^\infty_{t:t+N}$, $ \bar{\mathbf{u}}^\infty_{t:t+N-1}$, $ \bar{\boldsymbol{\lambda}}_{t}$ is optimal for \eqref{eq:OP_LMPC} at time $t$. Since we have already shown optimality of $\bar{\mathbf{x}}^\infty_{0:N}$, $\bar{\mathbf{u}}^\infty_{0:N}$, $\bar{\boldsymbol{\lambda}}_0$, we have that $\bar{\mathbf{x}}^\infty_{t:t+N}, \bar{\mathbf{u}}^\infty_{t:t+N}, \bar{\boldsymbol{\lambda}}_{t}$ is optimal for \eqref{eq:OP_LMPC} $\forall t\geq 0$, $\forall j \geq c+1$. Notice that $\bar{\mathbf{x}}^\infty_{0:N}$, $\bar{\mathbf{u}}^\infty_{0:N}$ is optimal for \eqref{eq:OP_LMPC} for any $N$, and $$\lim_{N\rightarrow\infty}(\bar{\mathbf{x}}^\infty_{0:N}\bar{\mathbf{u}}^\infty_{0:N}) = (\mathbf{x}^\infty, \mathbf{u}^\infty),$$
implying that the closed-loop trajectory can be equivalently obtained as the solution of the optimization problem \eqref{eq:OP_LMPC} as $N\rightarrow\infty$. For a given $N$, let us define the following sets of state-input trajectories and multipliers:
\begin{align*}
    \mathcal{F}^j_{1,N}&=\{\mathbf{x}^j_{0:N}, \mathbf{u}^j_{0:N-1}, \boldsymbol{\lambda}_0^j ~|~ \eqref{eq:stateconstraint}, \eqref{eq:initconstraint} \text{ hold}\}\\
    \mathcal{F}^j_{2,N}&=\{\mathbf{x}^j_{0:N}, \mathbf{u}^j_{0:N-1}, \boldsymbol{\lambda}_0^j~ |~ \eqref{eq:termconstraint}\text{ holds}\}\\
    \mathcal{F}^j_{3,N}&=\{\mathbf{x}^j_{0:N}, \mathbf{u}^j_{0:N-1}, \boldsymbol{\lambda}_0^j~ |~\eqref{eq:convhullconstraint},\eqref{newconstraint} \text{ hold}\}
\end{align*} 
Then using the optimality of $J^{\mathrm{MO-LMPC},j}_{0,N}(x^j_0)$ we have 
\begingroup
\allowdisplaybreaks
\begin{subequations}
    \begin{align}
    \lim_{N\rightarrow\infty}&J^{\mathrm{MO-LMPC},j}_{t,N}(x^j_0) \nonumber \\
    &\leq \min_{\substack{\tilde{\mathbf{x}}^j, \tilde{\mathbf{u}}^j, \tilde{\boldsymbol{\lambda}}^{j}_0\in\mathcal{F}^j_{m,\infty},\\
    \forall m=1,2,3}} \lim_{N\rightarrow\infty} \sum_{i=1}^M \alpha_i J^j_i(\tilde{\mathbf{x}}^{j}_{0:N}, \tilde{\mathbf{u}}^{j}_{0:N-1}, \tilde{\boldsymbol{\lambda}}^{j}_0)\\
    &=\min_{\substack{\tilde{\mathbf{x}}^j, \tilde{\mathbf{u}}^j, \tilde{\boldsymbol{\lambda}}^{j}_0\in\mathcal{F}^j_{m,\infty},\\
    \forall m=1,2,3}} \sum_{i=1}^M \alpha_i(\sum_{k=0}^{\infty} h_i(x_k^{j}, u_k^{j}) + \lim_{N\rightarrow\infty} \hat{V}_i^j(x_N^j, \tilde{\boldsymbol{\lambda}}^{j}_0))\label{eq:c2g_ub}\\
    &=\min_{\substack{\tilde{\mathbf{x}}^j, \tilde{\mathbf{u}}^j, \tilde{\boldsymbol{\lambda}}^{j}_0\in\mathcal{F}^j_{m,\infty},\\
    \forall m=1,3}} \sum_{i=1}^M \alpha_i\sum_{k=0}^{\infty} h_i(x_k^{j}, u_k^{j})\label{eq:strng_cvx}\\
    &=\min_{\substack{\tilde{\mathbf{x}}^j, \tilde{\mathbf{u}}^j, \tilde{\boldsymbol{\lambda}}^{j}_0\in\mathcal{F}^j_{1,\infty}}} \sum_{i=1}^M \bar{\alpha}_i\sum_{k=0}^{\infty} h_i(x_k^{j}, u_k^{j})\label{eq:lem}\\
    &= \min_{\substack{\boldsymbol{\pi}}}~~ \sum_{i=1}^M\bar{\alpha}_iV_i(x_0, \boldsymbol{\pi})\label{eq:equiv_MOopt} \\
    & ~~~~~\text{s.t}~~~ x_{t+1}=Ax_t+B\pi_t(x_t)\nonumber \\
    &~~~~~~~~~~~x_t\in\mathcal{X},~\pi_t(x_t) \in \mathcal{U} ~~\forall t\geq 0 \nonumber\\
    &~~~~~~~~~~~x_0=x_S \nonumber
\end{align}
\end{subequations}
\endgroup
The equality \eqref{eq:strng_cvx} is obtained by using the strong convexity of $h_{i'}(\cdot)$ and $\alpha_{i'}>0$ from Assumption \eqref{assmp:strng_cvx}--any optimal state-input trajectory of \eqref{eq:c2g_ub} must converge to a unique $x_F\in\mathcal{CS}^j$, and so $\lim_{N\rightarrow\infty}\hat{V}^j_i(x^j_N,\bar{\boldsymbol{\lambda}}^j_0) = \hat{V}^j_i(x_F, \bar{\boldsymbol{\lambda}}^j_\infty)=0$, and the constraint $\mathcal{F}_{2,\infty}$ is redundant. The equality constraint \eqref{eq:lem} is obtained by invoking Lemma \ref{lem:equiv_mo_prob}. The final equality constraint is obtained from the definition of $V_i(x_0, \boldsymbol{\pi})$. Notice that $\mathbf{x}^\infty, \mathbf{u}^\infty$ is feasible for \eqref{eq:equiv_MOopt} and also achieves the lower bound $\lim_{N\rightarrow\infty}J^{\mathrm{LMPC}}_{0,N}(x^j_0)$. Thus the policy $\boldsymbol{\pi}^j$ inducing the closed-loop trajectory $\mathbf{x}^\infty, \mathbf{u}^\infty$ $\forall j\geq c+1$ is the optimal solution to \eqref{eq:equiv_MOopt}, and is a Pareto optimal solution of \eqref{eq:inft_mo_opt} because of the convexity of \eqref{eq:equiv_MOopt} \cite{miettinen1999nonlinear}.  \hfill $\blacksquare$
\end{proof}

\section{Numerical Example}\label{sec:example}
We apply the MO-LMPC framework for the constrained optimal control of an LTI system.
Consider an agent described by the following discrete-time point mass model for motion on the 2D plane,
\begin{align}\label{eq:examplemodel}
\underbrace{\begin{bmatrix}X_{t+1}\\Y_{t+1}\\\dot{X}_{t+1}\\\dot{Y}_{t+1}\end{bmatrix}}_{x_{t+1}} = \begin{bmatrix}1&0&\Delta t & 0\\ 0 & 1 & 0 & \Delta t\\
0 & 0 & 1 & 0\\ 0 & 0 & 0 & 1\end{bmatrix}\underbrace{\begin{bmatrix}X_t\\Y_t\\\dot{X}_t\\\dot{Y}_t\end{bmatrix}}_{x_t} + \begin{bmatrix}\frac{1}{2}\Delta t^2& 0 \\
0 & \frac{1}{2}\Delta t^2\\
\Delta t& 0\\0&\Delta t\end{bmatrix}\underbrace{\begin{bmatrix}
    \ddot{X}_t\\\ddot{Y}_t
\end{bmatrix}}_{u_t},
\end{align}
where at time step $t$ the state $x_t$ comprises the agent's position $(X_t, Y_t)$ and velocity $(\dot{X}_t, \dot{Y}_t)$. The input $u_t$ is the acceleration in each direction $(\ddot{X}_t, \ddot{Y}_t)$, and the discretization time step is $\Delta t = 0.25\mathrm{s}$. The agent is subject to constraints 
\begin{align}\label{eq:exampleconstraints}
    \mathcal{X} = \{x~ | ~\mathrm{x}_{\mathrm{lb}} \leq x\leq  \mathrm{x}_\mathrm{ub}\},~ \mathcal{U} = \{u~ | ~\mathrm{u}_\mathrm{lb}\leq u \leq \mathrm{u}_\mathrm{ub}\},
\end{align}
where $\mathrm{x}_\mathrm{ub} = -\mathrm{x}_\mathrm{lb} = (4\mathrm{m}, 4\mathrm{m}, 4\mathrm{ms^{-1}}, 4\mathrm{ms^{-1}})$ and $\mathrm{u}_\mathrm{ub} = -\mathrm{u}_\mathrm{lb} = (1\mathrm{ms^{-2}}, 1\mathrm{ms^{-2}})$.  

The iterative control task is to steer the agent from the initial state $x_S = (4\mathrm{m}, 3\mathrm{m}, 0\mathrm{ms^{-1}}, 0\mathrm{ms^{-1}})$ to standstill at the origin $x_F = (0, 0, 0, 0)$. Control performance is evaluated with respect to two distinct objectives:
\begin{align}
V_1(x_S, \boldsymbol{\pi}) &= \sum_{t\geq 0 } x_t^\top Q x_t + \pi(x_t)^\top R \pi(x_t),~~\quad \quad[\text{Time to Target}]\label{eq:stabcost_eg}\\
V_2(x_S, \boldsymbol{\pi}) &= \sum_{t\geq 0 } \eta\Delta t \sqrt{\dot{X}_t^2 + \dot{Y}_t^2},~~~~~\quad \quad [\text{Energy Consumption}]\label{eq:soccost_eg}
\end{align}
where  $Q = \mathrm{diag}((0.5, 0.5, 0.02, 0.02)), R=\mathrm{diag}((0.6, 0.6))$ and $\eta = 12.0\mathrm{m^{-1}}$. The first objective \eqref{eq:stabcost_eg} is a proxy to approximately penalize the time taken to reach the task target at the origin ($x = 0,~ u = 0$) using a quadratic cost for strong convexity required by Assumption~\ref{assmp:strng_cvx}. (Note an additional objective $V_3(x_S,\boldsymbol{\pi}) = \sum_{t\geq 0} \Delta t$ can be introduced to penalize time exactly, but is omitted in this study for ease of visualization.) The second objective \eqref{eq:soccost_eg} penalizes the agent's energy consumption during the task. The agent is assumed to have a battery that is initially fully charged, with 100\% state of charge (SOC) at $t=0$. While energy consumption is generally a function of both speed and acceleration \cite{joa2024ecodriving}, here we assume a simpler model where the battery drains proportionally with the agent's speed:
\begin{align}\label{eq:examplesocstate}
    SOC_{t+1} = SOC_t - \eta \Delta t \sqrt{\dot{X}^2_t+\dot{Y}_t^2}.
\end{align}
Note $SOC_t$ is not an explicit state in our model \eqref{eq:examplemodel}.

The initialization trajectory for $\mathcal{CS}^0$ \eqref{eq:CS_def} is created using an MPC policy which models the system states as in \eqref{eq:examplemodel}  but with the additional SOC state \eqref{eq:examplesocstate} bounded via the nonlinear constraints $0 \leq SOC_k \leq 100,~\forall k=t,\dots, t+N$, resulting in a Nonlinear MPC formulation.
The objective is given by \eqref{eq:stabcost_eg}, with horizon $N=20$, and the problem is solved using IPOPT \cite{biegler2009large}.

After initializing $\mathcal{CS}^0$, we solve the multi-objective control task using the MO-LMPC formulation presented in Section~\ref{sec:MOLMPC}. The MO-LMPC is formulated as in \eqref{eq:OP_LMPC}, with the objective functions \eqref{eq:stabcost_eg} and \eqref{eq:soccost_eg}, 
 $\alpha = (0.4, 0.6)$, and MPC horizon $N=5$. The convex optimization problem is solved using ECOS \cite{domahidi2013ecos} \footnote{Our code is available at \url{https://github.com/shn66/Multi-Objective-LMPC}}.

The closed-loop state trajectories and objectives are depicted for $26$ iterations in Figs. \ref{fig:closedloop_traj} and \ref{fig:closedloop_performance}. As shown in Thm.~\ref{thm:feas}, all closed-loop trajectories converge to $x_F$.
Figure~\ref{fig:closedloop_traj} compares the initialization trajectory with the final trajectory ($j=26$): during iteration $j=26$ the agent state converges to the origin in fewer time steps (minimizing \eqref{eq:stabcost_eg}) while consuming less energy (minimizing \eqref{eq:soccost_eg}). Figure~\ref{fig:closedloop_performance} shows that both control costs improve with each task iteration, as in Thm.~\ref{thm:cost}. 

The choice of $\alpha$ suggests a prioritization between control objectives. As seen in Fig.~\ref{fig:closedloop_performance}, increasing $\alpha_1$ from $0.01$ to $0.99$ results in trajectories that increasingly prioritize minimizing \eqref{eq:stabcost_eg}. 
(However, we note that the strict performance improvement seen in Fig.~\ref{fig:closedloop_performance} cannot be guaranteed via choice of time-invariant $\alpha$ alone; the constraint \eqref{newconstraint} is required.)
As shown in Thm.~\ref{thm:pareti}, by varying $\alpha$ the MO-LMPC also converges to different Pareto optimal solutions. This is supported by results in Fig.~\ref{fig:closedloop_performance}, where the Pareto frontier was obtained by varying $\alpha_1$ in the range $0.01 \leq \alpha_1 \leq 0.99$ and running MO-LMPC to convergence; all converged solutions are non-dominating.
Note that due to our requirement of iterative performance improvement in each control objective, parts of the Pareto frontier may not be reachable using MO-LMPC. 
While MO-LMPC always converges to a point on the Pareto frontier, if certain Pareto optimal policies induce a larger cost in any control objective than the MO-LMPC initialization trajectory, those policies are not feasible for the MO-LMPC formulation. 
Future work will consider the design of an outer-loop algorithm for relaxing the performance improvement condition to incorporate designer preferences in the converged Pareto optimal solution.


\begin{figure}[ht]
    \centering 
\includegraphics[width=1.1\columnwidth]{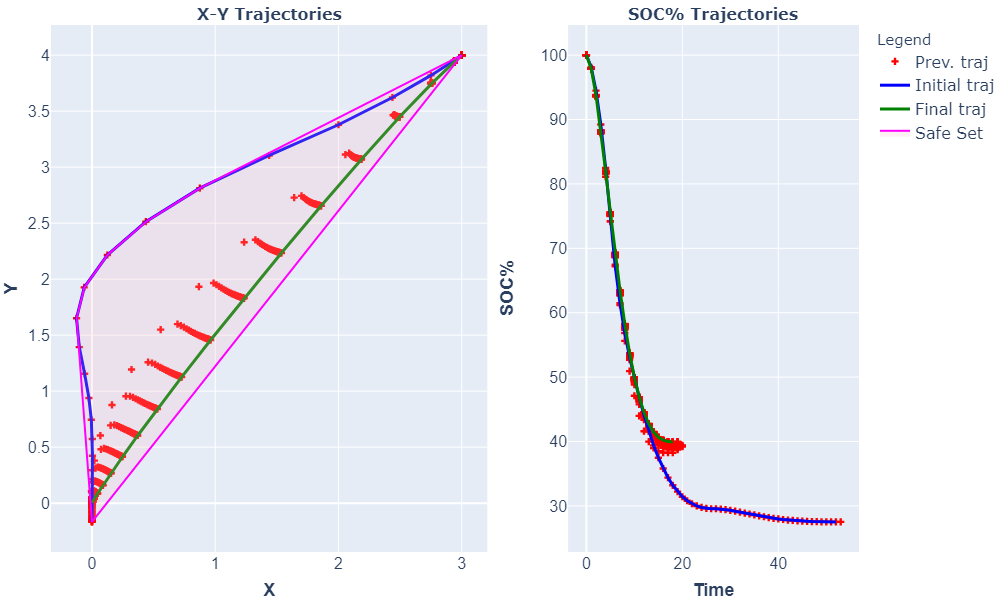}
    \caption{$X-Y$ and SOC\% trajectories of the system. Notice that the agent is stabilized sooner while consuming less SOC\% in iteration $j=26$, compared to iteration $j=0$.}
    \label{fig:closedloop_traj}
    
\end{figure}
\begin{figure}[ht]
    \centering   
    \includegraphics[width=1.1\columnwidth]{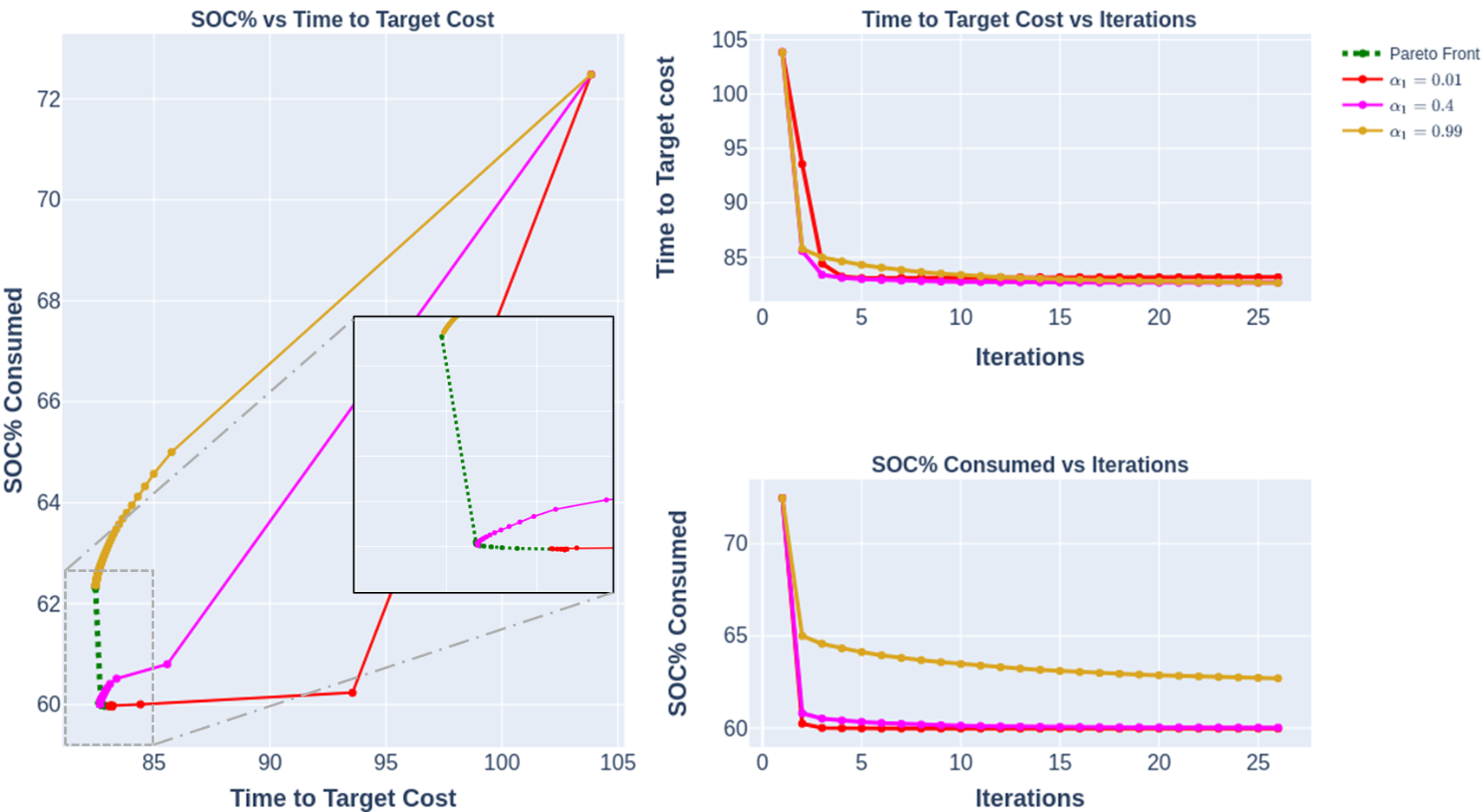}
    \caption{Trajectory costs \eqref{eq:stabcost_eg} and SOC\% consumed \eqref{eq:soccost_eg} across iterations, depicting iterative improvement in both objectives, and the effect of $\alpha$ on converged solutions.}
 \label{fig:closedloop_performance}
      
\end{figure}

\section{Conclusion}
This paper introduced the Multi-Objective Learning Model Predictive Controller (MO-LMPC), a control approach for iterative performance improvement with respect to multiple control costs.
We proved that the MO-LMPC controller is recursively feasible, asymptotically stable, iteratively reduces multiple control costs, and converges to a Pareto optimal solution to an infinite-horizon multi-objective optimal control problem.
The method was demonstrated via a numerical example for capacity-constrained control of a linear system. 
\bibliographystyle{ieeetr}
\bibliography{bib} 
%
%








\end{document}